\documentclass[11pt]{amsart}
\usepackage{fullpage}
\usepackage{amsthm,amsmath,amsfonts,amssymb} 
\usepackage{graphicx}

\newtheorem{thm}{Theorem}

\newtheorem{lma}[thm]{Lemma} 
\newtheorem{prop}[thm]{Proposition}

\newtheorem{rem}[thm]{Remark}

\def\bar{\overline}

\def\cO{\mathcal{O}}

\def\gf{\mathrm{gf}}
\def\gh{\textup{gh}}

\def\half{\tfrac{1}{2}}

\DeclareMathOperator{\ord}{ord}

\DeclareMathOperator{\tr}{Tr}
\def\tilde{\widetilde}

\renewcommand{\phi}{\varphi}

\title{Renormalization of the spectral action\\ for the Yang--Mills system}
\author{Walter D. van Suijlekom}
\address{Institute for Mathematics, Astrophysics and Particle Physics,
Radboud University Nijmegen, Heyendaalseweg 135, 6525 AJ Nijmegen, The Netherlands}
\email{waltervs@math.ru.nl}
\date{8 February 2011}

\begin{document}

\bibliographystyle{plainmath}
\begin{abstract}
We establish renormalizability of the full spectral action for the Yang--Mills system on a flat 4-dimensional background manifold. Interpreting the spectral action as a higher-derivative gauge theory, we find that it behaves unexpectedly well as far as renormalization is concerned. Namely, a power counting argument implies that the spectral action is superrenormalizable. From BRST-invariance of the one-loop effective action, we conclude that it is actually renormalizable as a gauge theory.
\end{abstract}

\maketitle

\section{Introduction}
One of the great successes of noncommutative geometry \cite{C94} is in its application to high-energy physics. Replacing the spacetime manifold by a noncommutative manifold, one puts the full Standard Model of elementary particles on equal geometrical footing as Einstein's General theory of Relativity. This is worked out in full detail in \cite{CCM07} (see also \cite{CM07} and the companion \cite{CC10}), including the physical predictions that are a consequence of this description.

Being a geometrical description of the Standard Model that is comparable to General Relativity makes it immediate that its quantization comes with the usual problems, actually typical for the latter theory. At the moment, one works with the noncommutative manifold as setting the {\it classical} starting point -- indeed allowing for a derivation of the full Standard Model Lagrangian at the classical level. Then, one adopts the physics textbook perturbative quantization of it, and arrive at physical predictions via the known Standard Model RG-equations. It needs no stressing that the situation around its quantization should be improved, and in the present letter we intend to take a first step in this direction.

We start with the full asymptotic expansion of the spectral action of Chamseddine and Connes \cite{CC96,CC97} in the case of the Yang--Mills system on a flat background manifold. 
By naive power counting we show -- after a suitable gauge-fixing -- that the full spectral action is superrenormalizable as a higher-derivative gauge theory \cite{Sla71,Sla72b} (cf. \cite[Section 4.4]{FS80}). Then, we demonstrate that the needed counterterms are gauge invariant polynomials that can safely be added to the spectral action. This shows renormalizability of the full spectral action for the Yang--Mills action, compatibly with gauge invariance.

\section{The Yang--Mills system}
The object of study in this paper is the spectral action for the Yang--Mills (YM) system on a flat background manifold. It is given by the relatively simple formula:
$$
S[A] := \tr f(D_A/\Lambda).
$$
This {\it spectral action} has firm roots in the noncommutative geometrical description of the Yang--Mills system, we refer to \cite{CCM07} for more details. For our purposes, it suffices to know that $D_A$ is a Dirac operator with coefficients in a $SU(N)$-vector bundle equipped with a connection $A$. That is, locally we have 
$$
D_A = i \gamma^\mu  \nabla_\mu + \gamma^\mu A_\mu.
$$
with $\nabla_\mu$ the spin connection on a Riemannian spin manifold $M$. For simplicity, we take $M$ to be flat ({\it i.e.} vanishing Riemann curvature tensor) and 4-dimensional. 
Furthermore, we will assume that $f$ is a Laplace transform:
$$
f(x) = \int_{t>0} e^{-tx^2} g(t) dt,
$$
even though this assumption could be avoided by using spectral densities instead (\cite{EGV98} and also \cite[Section 8.4]{Var06})
\begin{prop}[\cite{CC97}]
In the above notation, there is an asymptotic expansion (as $\Lambda \to \infty$):
\begin{equation}
\label{sa-eym}
S[A]  \sim \sum_{m \geq 0} \Lambda^{4-m} f_{4-m} \int_M a_m (x,D_A^2),
\end{equation}
in terms of the Seeley--De Witt invariants of $D_A^2$.
The coefficients are defined by $f_k := \int t^{-k/2} g(t)dt$.
\end{prop}
Recall that the Seeley--De Witt coefficients $a_m(x,D_A^2)$ are gauge invariant polynomials in the fields $A_\mu$. Indeed, the Weitzenb\"ock formula gives
$$
D_A^2 =- (\partial_\mu - i A_\mu) (\partial^\mu - i A^\mu)+ i\sum_{\mu < \nu}\gamma_\mu \gamma_\nu F_{\mu\nu}
$$
in terms of the curvature $F_{\mu\nu} = \partial_\mu A_\nu - \partial_\nu A_\mu -i [A_\mu,A_\nu]$ of $A_\mu$. Consequently, a Theorem by Gilkey \cite[Theorem 4.8.16]{Gil84} shows that (in this case) $a_m$ are polynomial gauge invariants in $F_{\mu\nu}$ and its covariant derivatives. The {\it order} $\ord$ of $a_m$ is $m$, where we set on generators:
$$
\ord A_{\mu_1; \mu_2\cdots \mu_k} = k.
$$
Consequently, the curvature $F_{\mu\nu}$ has order $2$, and $F_{\mu_1 \mu_2; \mu_3 \cdots \mu_k}$ has order $k$. For example, $a_4(x,D_A^2)$ is proportional to $\tr F_{\mu\nu}F^{\mu\nu}$ and more generally:
$$
a_{4+2k}(x,D_A^2) =c_k \tr F_{\mu\nu} \Delta^k_A (F^{\mu\nu}) + \cO(F^3)
$$
for some constants $c_k$ and the Laplacian $\Delta_A= -(\partial_\mu - i A_\mu)^2$ (see also \cite{Avr99} and references therein). 
The remainder is of third and higher order in $F$, plus its covariant derivatives, adding up to an order equal to $4+2k$.

\begin{rem}
It is the term $a_4$ that gives rise to the Yang--Mills action functional, the higher-order terms are usually ignored (being proportional to an inverse power of the `cut-off' $\Lambda$). More recently, also the higher-order terms, or even the full spectral action were studied in specific cases in \cite{CC11,MPT10} and from a more general point of view in \cite{Sui10}.
\end{rem}

\begin{prop}
The quadratic term $S_0[A]$ in $S[A]$ is given asymptotically (as $\Lambda \to \infty$) by
$$
S_0[A] \sim \sum_{k \geq 0} \Lambda^{-2k} f_{-2k} c_k \int \tr \hat F_{\mu\nu} \Delta^{k} (\hat F^{\mu\nu})
$$
where we have set $\hat F_{\mu\nu} = \partial_\mu A_\nu - \partial_\nu A_\mu$ and $\Delta= -\partial^\mu \partial_\mu$.
\end{prop}
We assume that the first term is the usual (free part of the) Yang--Mills action, that is, we adjust the positive function $f$ so that $f_0 c_0 =-1/4$. For the other coefficients, we have the following neat expression.\footnote{The coefficients $f_{2k}$ for positive $k$ were found to be the $k+1$'th moments of $f$, cf. \cite[Sect. 1.11]{CM07} for more details.}
\begin{lma}
The coefficients $f_{-2k}$ are related to the $2k$'th derivatives of $f$ at zero:
$$
f_{-2k} = \frac{(-1)^k f^{(2k)}(0)}{(2k-1)!!}.
$$
\end{lma}
\begin{proof}
With $f(x) = \int e^{-tx^2} g(t)dt$ we derive for its derivatives:
$$
f^{(2k)}(x) = \int_{t>0} e^{-tx^2/2} H_{2k}(\sqrt t x) t^{k} g(t)dt
$$
in terms of the Hermite polynomials $H_{n}(x) \equiv (-1)^n e^{x^2/2} (d/dx)^n e^{-x^2/2}$. Evaluating both sides at zero gives the desired result, using in addition that $H_{2k}(0)= (-1)^k (2k-1)!!$.
\end{proof}

We end this section by introducing a formal expansion
$
\phi_\Lambda(\Delta) = (f_0c_0)^{-1} \sum_{k \geq0} \Lambda^{-2k} f_{-2k} c_k  \Delta^k
$ (starting with $1$)
so that we can write more concisely
$$
S_0[A] \sim  -\frac{1}{4} \int \tr \hat F_{\mu\nu} \phi_\Lambda(\Delta) (\hat F^{\mu\nu})
$$
This form motivates the interpretation of $S_0[A]$ (and of $S[A]$) as a higher-derivative gauge theory. As we will see below, this indeed regularizes the theory in such a way that $S[A]$ defines a superrenormalizable field theory.

\section{Gauge fixing in the YM-system}

We add a gauge-fixing term of the following higher-derivative form:
\begin{equation}
\label{sa-gf}
S_\gf[A] \sim - \frac{1}{2 \xi}  \int \partial_\mu A^\mu \phi_\Lambda(\Delta) \left( \partial_\nu A^\nu \right) 
\end{equation}
We derive the {\it propagator} by inverting the non-degenerate quadratic form given by $S_0[A] + S_\gf[A]$:
$$
D_{\mu\nu}^{ab}(p; \Lambda) = \left[ g_{\mu\nu} - (1-\xi) \frac{p_\mu p_\nu}{ (p^2 +i \eta)}\right] \frac{\delta^{ab}}{(p^2 +i \eta) \phi_\Lambda(p^2)} 
$$
which for the moment is a formal expansion in $\Lambda$. We will come back to it in more detail in the next section.

As usual, the above gauge fixing requires a Jacobian, conveniently described by a Faddeev--Popov ghost Lagrangian:
\begin{equation}
\label{sa-gh}
S_\gh[A,\bar C,C] \sim - \int \partial_\mu \bar C \phi_\Lambda(\Delta) \left( \partial^\mu C + [A^\mu,C] \right)
\end{equation}
Here $C,\bar C$ are the Faddeev--Popov ghost fields and their propagator is
$$
\tilde D^{ab}(p; \Lambda) = \frac{\delta^{ab}}{(p^2 + i \eta) \phi_\Lambda(p^2)}.
$$

\begin{prop}
The sum $S[A] + S_\gf[A] + S_\gh[A,\bar C, C]$ is invariant under the BRST-transformations:
\begin{gather}
\label{brst}
sA_\mu = \partial_\mu C + [A_\mu,C];\qquad s C = -\half [C,C]; \qquad s \bar C =  \xi^{-1} \partial_\mu A^\mu.
\end{gather}
\end{prop}
\begin{proof}
First, $s(S)=0$ because of gauge invariance of $S[A]$. We compute 
$$
s(S_\gf) = -\frac{1}{\xi} \int ( \partial_\mu A^\mu) \phi_\Lambda(\Delta) \left( \partial_\nu \partial^\nu C + \partial_\nu( [A^\nu,C] \right) 
$$
On the other hand,
$$
s(S_\gh) = - \frac{1}{\xi} \int (\partial_\mu \partial^\nu A_\nu )\phi_\Lambda(\Delta) \left( \partial^\mu C+ [A^\mu,C] \right)
$$
which modulo vanishing boundary terms is minus the previous expression.
\end{proof}
Note that $s^2 \neq 0$, which can be cured by standard homological methods: introduce an auxiliary (aka Nakanishi-Lautrup) field $h$ so that $\bar C$ and $h$ form a contractible pair in BRST-cohomology. In other words, we replace the above transformation in \eqref{brst} on $\bar C$ by $s \bar C = --h$ and $s h = 0$. If we replace $S_\gf + S_\gh$ by $s \Psi$ with $\Psi$ an arbitrary {\it gauge fixing fermion}, it follows from gauge invariance of $S$ and nilpotency of $s$ that $S + s \Psi$ is BRST-invariant. The above special form of $S_\gf+ S_\gh$ can be recovered by choosing
$$
\Psi =-  \int \phi_\Lambda(\Delta) (\bar C) \left( \half \xi h + \partial_\mu A^\mu \right).
$$

\begin{rem}
One might wonder what gauge fixing condition is implemented by $S_\gf$ as in \eqref{sa-gf}, given the presence of the term $\phi_\Lambda(\Delta)$. Under suitable conditions on the function $f$, the function $x \mapsto \phi_\Lambda(x)$ is positive, turning the bilinear form
$$
(\omega_1,\omega_2) := - \int \tr \omega_1 \wedge \ast (\phi_\Lambda(\Delta) \omega_2)
$$
into an inner product. On the Lagrangian level, we can equally well implement the Lorenz gauge fixing condition $\partial \cdot A = 0$ using this inner product instead of the usual $L^2$-inner product. This gives rise to $S_\gf[A] =  ( \partial \cdot A, \partial \cdot A)/2\xi$. Similarly, $S_\gf$ is given by the inner product $(\bar C, \partial_\mu C + [A_\mu,C])$.
\end{rem}

\section{Renormalization of the spectral action for the YM-system}
As said, we consider the spectral action for the Yang--Mills system as a higher-derivative field theory. This means that we will use the higher derivatives of $F_{\mu\nu}$ that appear in the asymptotic expansion as natural regulators of the theory, similar to \cite{Sla71,Sla72b} (see also \cite[Sect. 4.4]{FS80}). However, note that the regularizing terms are already present in the spectral action $S[A]$ and need not be introduced as such.
Let us consider the expansion \eqref{sa-eym} up to order $n$ (which we assume to be at least $8$), {\it i.e.} we set $f_{4-m} = 0$ for all $m > n$. Also, assume a gauge fixing of the form \eqref{sa-gf} and \eqref{sa-gh}.

Then, we easily derive from the structure of $\phi_\Lambda(p^2)$ that the propagators of both the gauge field and the ghost field behave as $|p|^{-n+2}$ as $|p| \to \infty$. Indeed, in this case:
$$
\phi_\Lambda(p^2) = \sum_{k=0}^{n/2-2} \Lambda^{-2k} f_{-2k} c_k p^{2k}.
$$

Moreover, the weights of the interaction in terms of powers of momenta is given by:
$$
\begin{array}{|c|c|c|}
\hline
\text{vertex} & \text{valence} & \max \# \text{ der}\\
\hline
\hline
\parbox{2cm}{\vspace{1mm}\includegraphics[scale=.2]{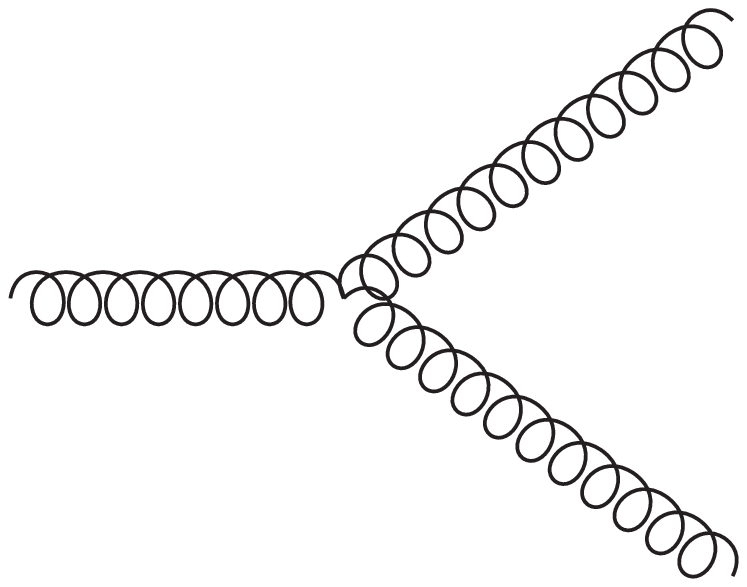}} & 3 & n-3\\[2mm]
\parbox{2cm}{\includegraphics[scale=.2]{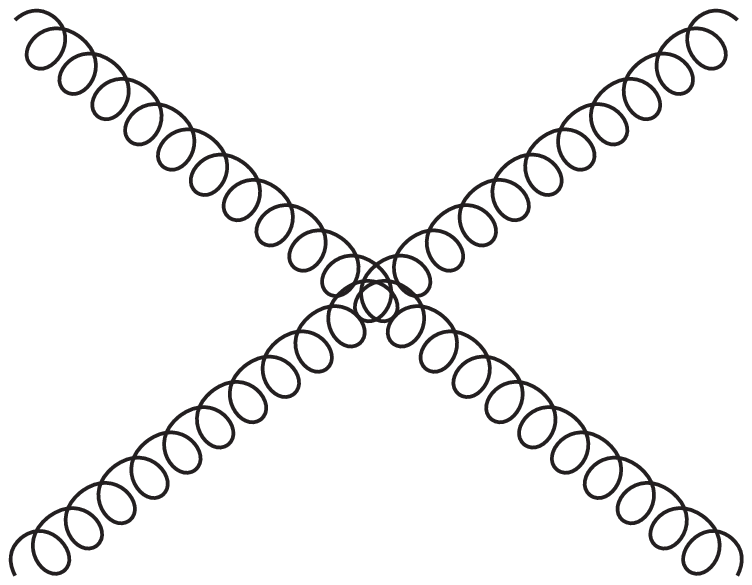}} & 4 & n-4\\[2mm]
\vdots& \vdots & \vdots\\
\parbox{2cm}{\includegraphics[scale=.2]{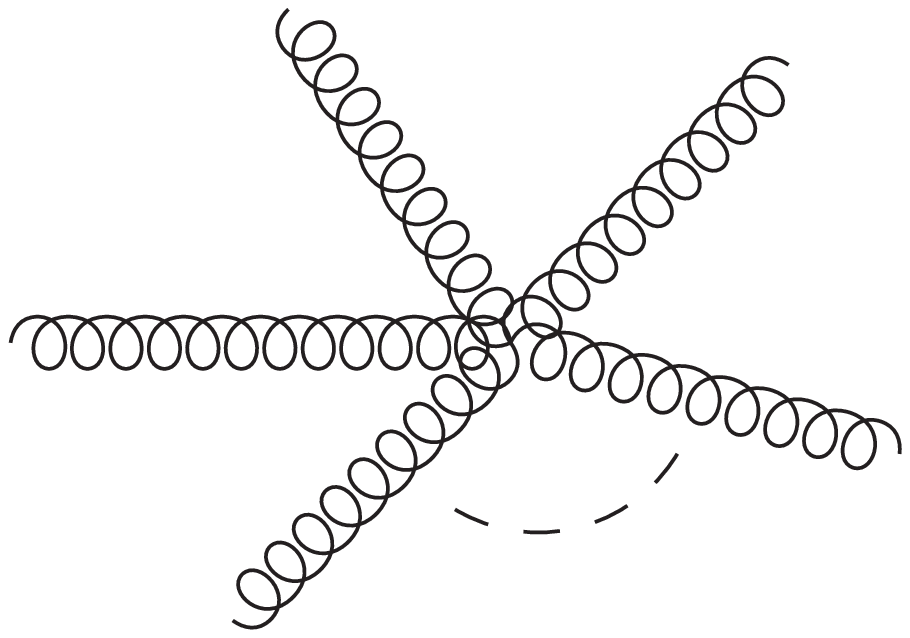}} & n & 0 \\[2mm]
\parbox{2cm}{\includegraphics[scale=.2]{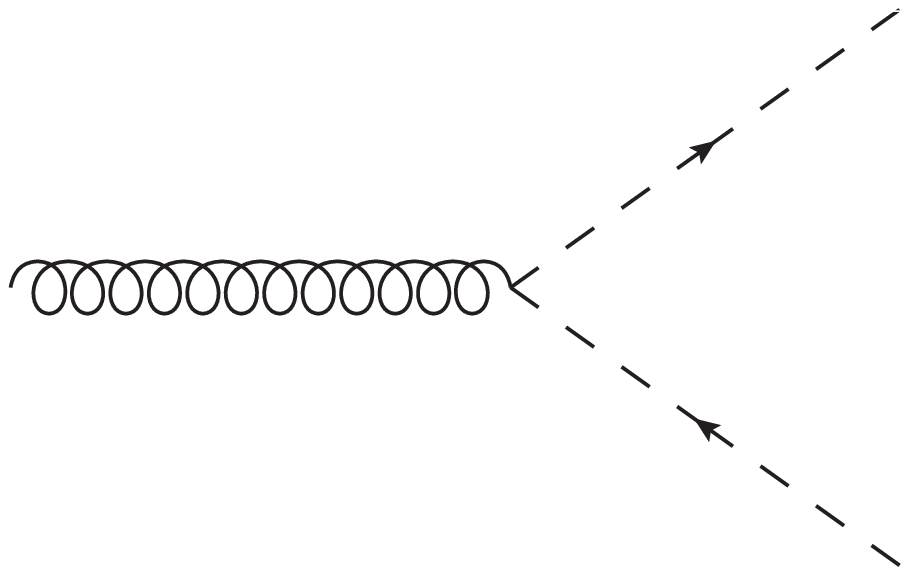}} & 3 & n-3\\[2mm]
\hline
\end{array}
$$
We will use $v_k$ to indicate the number of gauge interaction vertices of valence $k$, and with $\tilde v$ the number of ghost-gauge vertices.

Let us now find an expression for the {\it superficial degree of divergence} $\omega$ of a graph consisting of $I$ internal gauge edges, $\tilde I$ internal ghost edges, $v_k$ valence $k$ gauge vertices and $\tilde v$ ghost-gauge vertices. In 4 dimensions, we find at loop order $L$:
$$
\omega = 4L - I(n-2) - \tilde I (n-2) + \sum_{i=3}^n v_i (n-i) + \tilde v (n-3).
$$
\begin{lma}
Let $E$ and $\tilde E$ denote the number of external gauge and ghost edges, respectively. The superficial degree of divergence of the graph equals
$$
\omega = (4-n)(L-1) + 4 - (E+\tilde E).
$$
\end{lma}
\begin{proof}
We use the relations
$$
2 I + E = \sum_i i v_i + \tilde v; \qquad
2 \tilde I + \tilde E = 2\tilde v 
$$
where $E$ and $\tilde E$ are the number of external gauge and ghost legs, respectively. Indeed, these formulas count the number of half (gauge/ghost) edges in a graph in two ways: from the number of edges and from the valences of the vertices. We use them to substitute for $2I$ and $2\tilde I$ in the above expression for $\omega$ so as to obtain
$$
\omega = 4L - In - \tilde I n + n \left(\sum_i v_i + \tilde v \right) - (E + \tilde E)
$$
from which the result follows at once from Euler's formula $L= I + \tilde I - \sum_i v_i - \tilde v -1$.
\end{proof}
As a consequence, $\omega < 0$ if $L \geq 2$ (provided $n \geq 8$): all Feynman graphs are finite at loop order greater than 1. If $L=1$, then there are finitely many graphs which are divergent, namely those for which $E+ \tilde E \leq 4$. We conclude that the spectral action for the Yang--Mills system is superrenormalizable. 

Of course, the spectral action being a gauge theory, there is more to renormalizability than just power counting: we have to establish gauge invariance of the counterterms. 
We already know that the counterterms needed to render the perturbative quantization of the spectral action finite are of order $4$ or less in the fields and arise only from one-loop graphs. The key property of the effective action at one loop is that it is BRST-invariant:
$$
s(\Gamma_{1}) = 0. 
$$
In particular, assuming a regularization compatible with gauge invariance, the divergent part $\Gamma_{1,\infty}$ is BRST-invariant. Results from \cite{Dix91, DTV85, DTV85b, BDK90,DHTV91} on BRST-cohomology for Yang--Mills type theories ascertain that the only BRST-closed functional of order 4 or less in the fields is represented by
$$
\delta Z \int F_{\mu\nu}F^{\mu\nu}
$$
for some constant $\delta Z$. The counterterm $\Gamma_{1,\infty}$ can thus be added to $S$ and absorbed by a redefinition of the fields and coupling constant:
\begin{gather*}
A_0 = \sqrt{1+\delta Z} A ; \qquad g_0 = \frac{g}{\sqrt{1+ \delta Z}}
\end{gather*}
Equivalently, one could leave $A$ and $g$ invariant, and redefine $f_0 \mapsto (1+\delta Z) f_0$, leaving all other coefficients $f_{4-m}$ invariant. Intriguingly, renormalization of the spectral action for the YM-system can thus be accomplished merely by shifting the function $f$ in such a way that $f(0) \mapsto (1+\delta Z) f(0)$, whilst leaving all its higher derivatives at $0$ invariant. 

\begin{rem}
The above form for $\Gamma_{1,\infty}$ can actually be established by an explicit computation in dimensional regularization following \cite{PS96, PS97}. We intend to present the full details  elsewhere.
\end{rem}

\section*{Acknowledgements}
I would like to thank Klaas Landsman for useful discussions and remarks. I am grateful to Dmitri Vassilevich for useful comments. NWO is acknowledged for support under VENI-project 639.031.827.


\begin{thebibliography}{10}

\bibitem{Avr99}
I.~Avramidi.
\newblock {Covariant techniques for computation of the heat kernel}.
\newblock {\em Rev.Math.Phys.} 11 (1999)  947--980.

\bibitem{BDK90}
F.~Brandt, N.~Dragon, and M.~Kreuzer.
\newblock Lie algebra cohomology.
\newblock {\em Nucl. Phys.} B332 (1990)  250.

\bibitem{CC96}
A.~H. Chamseddine and A.~Connes.
\newblock Universal formula for noncommutative geometry actions: {U}nifications
  of gravity and the standard model.
\newblock {\em Phys. Rev. Lett.} 77 (1996)  4868--4871.

\bibitem{CC97}
A.~H. Chamseddine and A.~Connes.
\newblock The spectral action principle.
\newblock {\em Commun. Math. Phys.} 186 (1997)  731--750.

\bibitem{CC10}
A.~H. Chamseddine and A.~Connes.
\newblock {Noncommutative Geometry as a Framework for Unification of all
  Fundamental Interactions including Gravity. Part I}.
\newblock {\em Fortsch. Phys.} 58 (2010)  553--600.

\bibitem{CC11}
A.~H. Chamseddine and A.~Connes.
\newblock {Noncommutative Geometry as a Framework for Unification of all
  Fundamental Interactions including Gravity. Part II}.
\newblock To appear.

\bibitem{CCM07}
A.~H. Chamseddine, A.~Connes, and M.~Marcolli.
\newblock {Gravity and the standard model with neutrino mixing}.
\newblock {\em Adv. Theor. Math. Phys.} 11 (2007)  991--1089.

\bibitem{C94}
A.~Connes.
\newblock {\em Noncommutative Geometry}.
\newblock Academic Press, San Diego, 1994.

\bibitem{CM07}
A.~Connes and M.~Marcolli.
\newblock {\em Noncommutative Geometry, Quantum Fields and Motives}.
\newblock AMS, Providence, 2008.

\bibitem{Dix91}
J.~A. Dixon.
\newblock {Calculation of BRS cohomology with spectral sequences}.
\newblock {\em Commun. Math. Phys.} 139 (1991)  495--526.

\bibitem{DTV85}
M.~Dubois-Violette, M.~Talon, and C.~M. Viallet.
\newblock {BRS} algebras: Analysis of the consistency equations in gauge
  theory.
\newblock {\em Commun. Math. Phys.} 102 (1985)  105.

\bibitem{DTV85b}
M.~Dubois-Violette, M.~Talon, and C.~M. Viallet.
\newblock Results on {BRS} cohomology in gauge theory.
\newblock {\em Phys. Lett.} B158 (1985)  231.

\bibitem{DHTV91}
M.~Dubois-Violette, M.~Henneaux, M.~Talon, and C.-M. Viallet.
\newblock {Some results on local cohomologies in field theory}.
\newblock {\em Phys. Lett.} B267 (1991)  81--87.

\bibitem{EGV98}
R.~Estrada, J.~M. Gracia-Bond{\'\i}a, and J.~C. V\'arilly.
\newblock On summability of distributions and spectral geometry.
\newblock {\em Commun. Math. Phys.} 191 (1998)  219--248.

\bibitem{FS80}
L.~Faddeev and A.~Slavnov.
\newblock {\em Gauge Fields. Introduction to Quantum Theory}.
\newblock Benjaming Cummings, 1980.

\bibitem{Gil84}
P.~B. Gilkey.
\newblock {\em Invariance theory, the heat equation, and the {A}tiyah-{S}inger
  index theorem}, volume~11 of {\em Mathematics Lecture Series}.
\newblock Publish or Perish Inc., Wilmington, DE, 1984.

\bibitem{MPT10}
M.~Marcolli, E.~Pierpaoli, and K.~Teh.
\newblock {The spectral action and cosmic topology}, 1005.2256.

\bibitem{PS96}
P.~I. Pronin and K.~V. Stepanyantz.
\newblock {One-loop effective action for an arbitrary theory}.
\newblock {\em Teor. Mat. Fyz.} 109 (1996)  215--231.

\bibitem{PS97}
P.~I. Pronin and K.~V. Stepanyantz.
\newblock {One-loop counterterms for higher derivative regularized
  Lagrangians}.
\newblock {\em Phys. Lett.} B414 (1997)  117--122.

\bibitem{Sla71}
A.~A. Slavnov.
\newblock {Invariant regularization of nonlinear chiral theories}.
\newblock {\em Nucl. Phys.} B31 (1971)  301--315.

\bibitem{Sla72b}
A.~A. Slavnov.
\newblock {Invariant regularization of gauge theories}.
\newblock {\em Teor. Mat. Fiz.} 13 (1972)  174--177.

\bibitem{Sui10}
W.~D. van Suijlekom.
\newblock Perturbations and operator trace functions.
\newblock To appear in {\it J. Funct. Anal.}

\bibitem{Var06}
J.~C. V\'arilly.
\newblock {\em {An Introduction to Noncommutative Geometry}}.
\newblock European Math. Soc. Publishing House (EMS Series of Lectures in
  Mathematics), 2006.

\end{thebibliography}
\newcommand{\noopsort}[1]{}\def\cprime{$'$}

\end{document}